\documentclass[11pt,a4paper,reqno]{amsproc}

\input{header2.sty}

\begin{document}

\title[Stability of the enhanced area law]{Stability of the enhanced area law of the\\ entanglement entropy}

\author[P.\ M\"uller]{Peter M\"uller}
\address[P.\ M\"uller]{Mathematisches Institut,
  Ludwig-Maximilians-Universit\"at M\"unchen,
  Theresienstra\ss{e} 39,
  80333 M\"unchen, Germany}
\email{mueller@lmu.de}

\author[R.\ Schulte]{Ruth Schulte}
\address[R.\ Schulte]{Mathematisches Institut,
  Ludwig-Maximilians-Universit\"at M\"unchen,
  Theresienstra\ss{e} 39,
  80333 M\"unchen, Germany}
\email{schulte@math.lmu.de}
\thanks{Ruth Schulte was funded by the Deutsche Forschungsgemeinschaft (DFG, German Research Foundation) under Germany’s Excellence Strategy -- EXC-2111 -- 390814868.}

\begin{abstract}
	We consider a multi-dimensional continuum Schr\"odinger operator which is given by a perturbation 
	of the negative Laplacian by a compactly supported potential. 
	We establish both an upper and a lower bound on the bipartite 
	entanglement entropy of the ground state of the corresponding quasi-free Fermi gas.
	The bounds prove that the scaling behaviour of the entanglement entropy remains a logarithmically 
	enhanced area law as in the unperturbed case of the free Fermi gas. The central idea for the upper bound 
	is to use a limiting absorption principle for such kinds of Schr\"odinger operators. 
\end{abstract}

\maketitle
\thispagestyle{empty}

%%%%%%%%%%%%%%%%%%%%%%%%%%%%%%%%%%%%%%%%%%%%%%%%%%%%%%%%%%
%%%%%%%%%%%%%%%%%%%%%%%%%%%%%%%%%%%%%%%%%%%%%%%%%%%%%%%%%%
%
\section{Introduction and Result}
%
%%%%%%%%%%%%%%%%%%%%%%%%%%%%%%%%%%%%%%%%%%%%%%%%%%%%%%%%%%
%%%%%%%%%%%%%%%%%%%%%%%%%%%%%%%%%%%%%%%%%%%%%%%%%%%%%%%%%%

Entanglement properties of the ground state of quasi-free Fermi gases have received considerable attention 
over the last two decades, see e.g.\ \cite{Botero:2004kp, Keating:2004bd, Wolf:2006ek, Gioev:2006cl, Helling:2011gr, LeschkeSobolevSpitzer14, PhysRevLett.113.150404, AbdulRhamanStolz15, ElgartPasturShcherbina2016, LeschkeSobolevSpitzer17, MR3671049, PfirschSobolev18, MuPaSc19, CharlesEst20, Spitzer20}. Here, entanglement is understood with respect to a spatial bipartition of the system 
into a subsystem of linear size proportional to $L$ and the complement. 
Entanglement entropies are a common measure for entanglement. Often, the von Neumann entropy of the reduced ground state of the Fermi gas is considered. Its investigations give rise to non-trivial mathematical questions and 
to answers that are of physical relevance. 
This is true even for the simplest case of a quasi-free Fermi gas, namely the free Fermi gas with 
(single-particle) Hamiltonian $H_{0}:= -\Delta$ given by the Laplacian in $d \in \N$ space dimensions.
Its entanglement entropy was suggested \cite{Wolf:2006ek, Gioev:2006cl, 10.1155/IMRN/2006/95181, Helling:2011gr} to obey a \emph{logarithmically enhanced area law}, 
\begin{equation}
	\label{eq:widom}
 	S_{E}(H_{0},\Lambda_{L}) = \Sigma_{0} \,L^{d-1} \ln L + {o}\big(L^{d-1} \ln L\big),
\end{equation}
as $L\to\infty$. Here, $E>0$ stands for the Fermi energy, which characterises the ground state, and 
$\Lambda_L:=L \cdot \Lambda$ is the scaled version of some ``nice'' bounded subset $\Lambda \subset \R^{d}$, 
which is specified below in Assumption~\ref{ass:l}(i).
The leading-order coefficient 
\begin{equation}
	\label{eq:sigma0}
 	\Sigma_{0} \equiv \Sigma_{0}(\Lambda,E) := \frac{ E^{(d-1)/2}\, |\partial\Lambda|}{3 \cdot 2^{d}\pi^{(d-1)/2} \Gamma\big((d+1)/2\big)}, 
\end{equation}
where $|\partial\Lambda|$ denotes the surface area of the boundary $\partial\Lambda$ of $\Lambda$, was expected 
\cite{Gioev:2006cl, 10.1155/IMRN/2006/95181, Helling:2011gr} to be determined by Widom's 
conjecture \cite{Widom:1982uz}. 
This was finally proved in \cite{LeschkeSobolevSpitzer14} based on celebrated works by Sobolev \cite{MR3076415, Sobolev:2014ew}.
The occurrence of the logarithm $\ln L$ in the leading term of \eqref{eq:widom} is attributed to the delocalisation or transport properties of the Laplacian dynamics. It leads to long-range correlations in the ground state of the Fermi gas across the surface of the subsystem in $\Lambda_{L}$. 
If a periodic potential is added to $H_{0}$, and the Fermi energy falls into a spectral band, the logarithmically enhanced area law \eqref{eq:widom} is still valid, as was proven in \cite{PfirschSobolev18} for $d=1$. 

If $H_{0}$ is replaced by another Schr\"odinger operator $H$ with a mobility gap in the spectrum and if the Fermi energy falls into the mobility gap, then the $\ln L$-factor is expected to be absent in the leading asymptotic term of the entanglement entropy. Such a phenomenon is referred to as an \emph{area law}, namely
	$S_{E}(H,\Lambda_{L}) \sim L^{d-1}$ as $L\to\infty$.
 It was first observed by Bekenstein \cite{PhysRevD.7.2333, Bekenstein04} in a toy model for the Hawking entropy of black holes. 
 An area  law also holds if $H$ models a particle in a constant magnetic field \cite{CharlesEst20, Spitzer20}.
Area laws are proven to occur for random Schr\"odinger operators and Fermi energies in the region of dynamical localisation \cite{PhysRevLett.113.150404,ElgartPasturShcherbina2016, MR3744386}. 	
The proofs rely on the exponential decay in space of the Fermi projection for $E$ in the region of complete localisation.
It should be pointed out that spectral localisation alone is not sufficient for the validity of an area law. This has been recently demonstrated \cite{MuPaSc19} for the random dimer model if the Fermi energy coincides with one of the critical energies where the localisation length diverges and dynamical delocalisation takes over. 

Due to the complexity of the problem, there does not exist a mathematical approach which allows to determine the leading behaviour of the entanglement entropy for general Schr\"odinger operators $H$. All that is known is what happens for the examples discussed above.
The experts in the field have conjectured for a decade that, given $H$ with a ``reasonable'' potential, a possibly occurring enhancement to the area law for $S_{E}(H,\Lambda_{L})$ should not be stronger than logarithmic. Even though no counterexamples are known so far, proving the conjecture turned out to be a very difficult task which has not been solved yet. As an aside, we mention that for interacting quantum systems, stronger enhancements to area laws than logarithmic are known in peculiar cases. In fact, spin chains ($d=1$) can be designed in such a way  as to realise any growth rate up to $L$ \cite{Ramirez:2014kp,MR3582443}. 

 In this paper we undertake a first step towards a proof of the conjecture. We establish an upper bound on the entanglement entropy
corresponding to $H=-\Delta + V$ which grows like $L^{d-1} \ln L$ as $L\to\infty$, provided the potential $V$ is bounded and has compact support. Compactness of the support is the crucial restriction of our result. It could be relaxed to having a sufficiently fast decay at infinity, but we have chosen not to focus on this for reasons of simplicity. The main technical input in our analysis is a limiting absorption principle for $H$. Since $H$ has absolutely continuous spectrum filling the non-negative real half-line, one expects $S_{E}(H,\Lambda_{L})$ to obey an enhanced area law for Fermi energies $E>0$. Therefore, a corresponding lower bound, which grows also like $L^{d-1} \ln L$ as $L\to\infty$, is of interest, too. These findings are summarised in Theorem~\ref{thm:main}, which is our main result.
		The proof of the upper bound is much more involved than that of the lower bound. Both bounds require the representation of the Fermi projection as a Riesz projection with the integration contour cutting through the continuous spectrum. Such a representation may be of independent interest. We prove it in the Appendix in a more general setting for operators for which a limiting absorption principle holds.

\bigskip

Let $H:=-\Delta+V$ be a densely defined Schr\"odinger operator in the Hilbert space $L^{2}(\R^{d})$ with bounded potential $V\in L^\infty(\R^d)$.
According to \cite{Klich06} there exists a trace formula for the entanglement entropy which we take as our definition
\begin{equation}\label{eq:EE}
		S_E(H,\Omega) :=\tr\big\{h(1_{\Omega}1_{<E}(H)1_{\Omega})\big\}.
\end{equation}
Here, $\Omega \subset\R^{d}$ is any bounded Borel set, we write $1_{A}$ for the indicator function of a set $A$ and, in abuse of notation, 
$1_{<E} := 1_{]-\infty,E[}$ for the Fermi function with Fermi energy $E\in\R$. 
We also introduced the entanglement-entropy function $h:\,[0,1]\rightarrow[0,1]$,
\begin{equation}
		h(\lambda):=-\lambda\log_{2} \lambda - (1-\lambda)\log_{2}(1-\lambda),
\end{equation}
and use the convention $0 \log_{2}0 :=0$ for the binary logarithm.

\begin{ass}
	\label{ass:l}
	We consider a bounded Borel set $\Lambda\subset\R^{d}$ such that 
	\begin{enumerate}
 	\item[(i)]  it is a Lipschitz domain with, if $d\ge 2$, a piecewise $C^{1}$-boundary, 
	\item[(ii)] the origin $0\in \R^{d}$ is an interior point of $\Lambda$.
	\end{enumerate}
\end{ass}

\begin{remark}
	Assumption~\ref{ass:l}(i) is taken from \cite{LeschkeSobolevSpitzer14} and guarantees the validity of the enhanced area law \eqref{eq:widom} for the free Fermi gas which is proven there,  see also \cite[Cond.\ 3.1]{LeschkeSobolevSpitzer17} for the notion of a Lipschitz domain. Assumption~\ref{ass:l}(ii) does not impose any restriction because it 
	can always be achieved by a translation of the potential $V$ in Theorem~\ref{thm:main}.
\end{remark}

We recall that $\Lambda_{L} = L\cdot \Lambda$. The main result of this paper is summarised in 

\begin{theorem}\label{thm:main}
	Let $\Lambda\subset\R^{d}$ be as in Assumption~\ref{ass:l} and let
	$V\in L^\infty(\R^d)$ have compact support. Then, for every Fermi energy $E>0$ there exist
	constants $\Sigma_{l} \equiv \Sigma_{l}(\Lambda,E) \in {}]0,\infty[$ and $\Sigma_{u} \equiv \Sigma_{u}(\Lambda,E,V) \in {}]0,\infty[$ such that
	\begin{equation}
		\Sigma_{l}  \le \liminf_{L\rightarrow\infty}\frac{S_E(H,\Lambda_L)}{L^{d-1}\ln L}
		\le \limsup_{L\rightarrow\infty}\frac{S_E(H,\Lambda_L)}{L^{d-1}\ln L} \le \Sigma_{u}.
	\end{equation}
\end{theorem}

%%%%%%%%%%%%%%%%%%%%%%%%%%%%%%%%%%%%%%%%%%%%%Remark
\begin{remarks}
	\item 
		The constant $\Sigma_{l}$ can be expressed in terms of the coefficient $\Sigma_{0}$ 
		in the leading term of the unperturbed 
		entanglement entropy $S_{E}(H_{0}, \Lambda_{L})$ for large $L$, cf.\ \eqref{eq:widom} and \eqref{eq:sigma0}.
		The explicit form 
		\begin{equation}
 			\Sigma_{l} =  \frac{3\Sigma_0}{2\pi^{2}},
		\end{equation}
		is derived in \eqref{sigmal-def}. 
	\item 
		If $d>1$,  the constant $\Sigma_{u}$  can also be expressed in terms of $\Sigma_{0}$. 
		According to \eqref{upper-prefinal} and \eqref{Differenz-2s}, we have
		\begin{equation}\label{eq:SIgmau}
			\Sigma_{u} = 2508\Sigma_0.
		\end{equation} 
		In particular, this constant is independent of $V$. The numerical prefactor in \eqref{eq:SIgmau} can be 
		improved by using the alternative approach described in Remark~\ref{ReM.Abkuerzung}.	
		In $d=1$ dimension, however, we only obtain a constant $\Sigma_u$ which also depends on $V$, 
		because there is an additional contribution from \eqref{Differenz-2s}.
	\item
		Pfirsch and Sobolev \cite{PfirschSobolev18} proved that the coefficient of the leading-order term of the enhanced area law is not altered by adding a periodic potential in $d=1$. Therefore we expect the 
		$V$-dependence of $\Sigma_{u}$ in $d=1$ to be an artefact of our method.
	\item 
		At negative energies there is at most discrete spectrum of $H$. Thus, if $E<0$ the Fermi function can be 
		smoothed out without changing the operator $1_{<E}(H)$. Therefore, the operator kernel of $1_{<E}(H)$ has fast 
		polynomial decay, and $S_{E}(H,\Lambda_{L}) = \mathcal O(L^{d-1})$ follows as in 
		\cite{PhysRevLett.113.150404,ElgartPasturShcherbina2016}. In other words, the growth of the entanglement 
		entropy is at most an area law. The same holds at $E=0$ because eigenvalues cannot accumulate from below at $0$
		due to the boundedness of $V$ and its compact support.		
%		The latter follows, e.g., from the min-max-principle by comparison with a spherical potential well whose depth is given by the infimum of $V$ and whose radius is large enough so that the support of $V$ is contained in the corresponding ball. 
	\item
		The stability analysis we perform in this paper requires only that the spatial domain $\Lambda$ is a 
		bounded measurable subset of $\R^{d}$ which has an interior point. The stronger assumptions we make 
		are to ensure the validity of Widom's formula for the unperturbed system as proven 
		in  \cite{LeschkeSobolevSpitzer14}.  
\end{remarks}

%%%%%%%%%%%%%%%%%%%%%%%%%%%%%%%%%%%%%%%%%%%%%%%%%%%%%%%%%%
%%%%%%%%%%%%%%%%%%%%%%%%%%%%%%%%%%%%%%%%%%%%%%%%%%%%%%%%%%
%
\section{Proof of Theorem~\ref{thm:main}}
%
%%%%%%%%%%%%%%%%%%%%%%%%%%%%%%%%%%%%%%%%%%%%%%%%%%%%%%%%%%
%%%%%%%%%%%%%%%%%%%%%%%%%%%%%%%%%%%%%%%%%%%%%%%%%%%%%%%%%%

We prove the upper bound of Theorem~\ref{thm:main} in Section~\ref{sec:upper} and the lower bound in Section~\ref{sec:lower}.
Section~\ref{sec:both} contains results needed for both bounds.

\subsection{Preliminaries}
\label{sec:both}

Our strategy is a perturbation approach which bounds the entanglement entropy of $H$ in terms of that of $H_{0}$
for large volumes. We estimate the function $h$ in~\eqref{eq:EE} according to 
	\begin{equation}
		\label{h-bounds}
 		g \le h \le -3g \log_{2}g,
	\end{equation}
where 
	\begin{equation}\label{eq:g}
		g:\,[0,1]\rightarrow[0,1], \; \lambda\mapsto \lambda(1-\lambda),
	\end{equation}
see Lemma \ref{lem:gs} for a proof of the lower bound in \eqref{h-bounds} 
and Lemma~\ref{lem:logquatrat} for a proof of the upper bound. Thus, we will be concerned with the operator 
	\begin{equation}
		\label{eq:gofProj}
		g\big(1_{\Lambda_L}1_{<E}(H_{(0)})1_{\Lambda_L}\big)
		=\big|1_{\Lambda_L^{c}}1_{<E}(H_{(0)})	1_{\Lambda_L}\big|^2,
	\end{equation}
where $|A|^2:=A^\ast A$ for any bounded operator $A$, and the superscript $^{c}$ indicates the complement of a set. 
This observation leads us to consider 
von Neumann--Schatten norms of operator differences $1_{\Lambda_L^{c}}[ 1_{<E}(H_{0}) -  1_{<E}(H)]1_{\Lambda_L}$, which is done in Lemma~\ref{lem:HilbertSchmidt} and Lemma~\ref{lem:interpolation}. Lemma~\ref{lem:HilbertSchmidt} allows to deduce the lower bound in Theorem~\ref{thm:main}, whereas the upper bound requires more work due to the presence of the additional logarithm. Lemma~\ref{lem:logtriangleineq} will tackle this issue. 

In order to show the crucial Lemma~\ref{lem:HilbertSchmidt}, we need two preparatory results. The first one is about the decay in space of the 
free resolvent in Lemma~\ref{lem:RoselventHilbertSchmidt}. 
For $z\in\C\backslash\R$ let $G_0(\,\boldsymbol\cdot\,,\,\boldsymbol\cdot\,;z):\;\R^d\times\R^d\rightarrow\C$ be the kernel of the resolvent $\frac{1}{H_0-z}$.
The explicit formula for $G_0(\,\boldsymbol\cdot\,,\,\boldsymbol\cdot\,;z)$ is well known. Likewise there exists an estimate for $G_0(\cdot,\cdot;z)$ evaluated for large arguments, i.e. there exists $R \equiv R(d)>0$ and $C \equiv C(d)>0$ such that for all $x,y\in\R^d$ with Euclidean distance $|x-y| \ge R/|z|^{1/2}$ we have
	\begin{equation}\label{eq:GreensfunctionEst}
		|G_0(x,y;z)|\le C |z|^{(d-3)/4} \,\frac{\e^{-|\Im \sqrt{z}||x-y|}}{|x-y|^{(d-1)/2}}.
	\end{equation}
For a reference, see \cite{MR262699} and \cite[Chap.\ 9.2]{MR0167642} for $d\ge 2$ and \cite[Chap.\ I.3.1]{AlGeHoeHo} for $d=1$.
Here, $\sqrt{\,\boldsymbol\cdot\,}$ denotes the principal branch of the square root.

We write $\Gamma_l:=l+ [ 0,1]^d$ for the closed unit cube translated by $l\in\Z^d$.

%%%%%%%%%%%%%%%%%%%%%%%%%%%%%%%%%%%%%%%%%%%%%%Lemma
\begin{lemma}\label{lem:RoselventHilbertSchmidt}
	Let $V\in L^\infty(\R^d)$ with compact support in $[-R_{V},R_{V}]^{d}$ for some $R_V >0$. 
	Given $z\in\C\setminus\R$, let $\ell_0\equiv \ell_{0}(d,V,z):= 2\sqrt{d}(R_V+1) + R(d)/|z|^{1/2}$. 
	
	Then, there exists a constant $C_1\equiv C_1(d,V)>0$ such that for any $z\in\C\setminus\R$ and any
	$n\in \Z^d\setminus {]}-\ell_{0},\ell_{0}[\,^{d}$ we have
	\begin{equation}\label{eq:4ResolvEst}
		\Big\||V|^{1/2}\frac{1}{H_0-z}1_{\Gamma_n}\Big\|_4\le C_1 |z|^{(d-3)/4} \; \frac{\e^{-|\Im \sqrt z||n|/2}}{|n|^{(d-1)/2}}.
	\end{equation}
	Here, $\|\boldsymbol\cdot\|_{p}$ denotes the von Neumann--Schatten norm for $p\in[1,\infty[$.
\end{lemma}

\begin{proof}
	Let $z\in\C\setminus\R$.
	Since the Hilbert--Schmidt norm of an operator can be computed in terms of the integral kernel, we get
	\begin{align}
		\Big\||V|^{1/2}\frac{1}{H_0-z}1_{\Gamma_n}\Big\|_4^{4} 
		&= \Big\|1_{\Gamma_n}\frac{1}{H_0-\overline{z}}\,|V|\,\frac{1}{H_0-z}1_{\Gamma_n}\Big\|_2^{2} \notag\\
		&=\int_{\Gamma_n}\dd x\int_{\Gamma_n}\dd y\, 
		\bigg|\int_{\R^{d}}\dd\xi\, G_0(x,\xi;\overline{z})\, |V(\xi)| \,G_0(\xi,y;z)\bigg|^2 .
	\end{align}
	For every $n\in \Z^d\setminus {}]-\ell_{0},\ell_{0}[\,^{d}$, every $x\in\Gamma_n$ and every $\xi\in\textrm{supp} V$, we 
	infer that $|x-\xi| \ge R(d)/|z|^{1/2}$. Therefore the Green's-function estimate \eqref{eq:GreensfunctionEst} yields
	\begin{equation}
		|G_0(x,\xi;z)|\le 2^{(d-1)/2} C(d)|z|^{(d-3)/4} \; \frac{\e^{-|\Im \sqrt z||n|/2}}{|n|^{(d-1)/2}}
	\end{equation}
	because 
	\begin{equation}
		|x - \xi| \ge |x| - \sqrt{d}R_V \ge |n|-\sqrt{d}(R_V+1) \ge \frac{|n|}{2}.
	\end{equation}
	This implies the lemma.
\end{proof}

As a second preparatory result for one of our central bounds we require

\begin{lemma}
	\label{lem:H-Riesz}
 	Let $V\in L^\infty(\R^d)$ with compact support.
	We fix an energy $E>0$ and consider two compact subsets $\Gamma,\Gamma' \subset\R^{d}$. 
	Then we have the representation
	\begin{equation}
		\label{eq:riesz-rep-H}
 		1_{\Gamma} 1_{<E}(H_{(0)}) 1_{\Gamma'} = - \frac{1}{2\pi\i} \oint_{\gamma} \d z\, 1_{\Gamma} \,
		\frac{1}{H_{(0)}-z} \, 1_{\Gamma'}.
	\end{equation}
	The right-hand side of \eqref{eq:riesz-rep-H} exists as a Bochner integral with respect to the operator norm, 
	and the integration contour $\gamma$ is a closed curve	in the complex plane $\C$ which traces the 
	boundary of the rectangle $\big\{z\in\C:\;|\Im z|\le E, \;\Re z\in[-1+ \inf\sigma(H),E] \big\}$
	once in the counter-clockwise direction.
\end{lemma}

\begin{proof}
 	The lemma follows from the corresponding abstract result in Theorem~\ref{thm:abs-Riesz} in the appendix. 
	Indeed, according to \cite[Thm.~4.2]{Agmon75}, see also e.g.\
	\cite{MR3650224}, both $H$ and $H_{0}$ fulfil a limiting absorption principle at any $E>0$,
	\begin{equation}
		\label{eq:lim-new}
		\sup_{z\in\C:\,\Re z=E, \, \Im z\neq 0} 
		\Big\|\langle X\rangle^{-1}\frac{1}{H_{(0)}-z}\Pi_{c}(H_{(0)})\langle X\rangle^{-1}\Big\| < \infty
	\end{equation}
	with $X$ being the position operator, $\langle\,\pmb\cdot\,\rangle:=\sqrt{1+|\pmb\cdot|^2}$ 
	the Japanese bracket and $\Pi_{c}(H_{(0)})$ the projection onto the continuous spectral subspace of
	$H_{(0)}$.
	Also, $\sigma_{pp}(H) \subset {}]-\infty,0]$ because the potential $V$ is bounded and compactly supported 
	\cite[Cor.\ on p.\ 230]{reed1980methods4}.
\end{proof}

The statement of the next lemma is a crucial estimate that will be needed for both the upper and the 
lower bound in Theorem~\ref{thm:main}.

%%%%%%%%%%%%%%%%%%%%%%%%%%%%%%%%%%%%%%%%%%%%%%%Lemma
\begin{lemma}\label{lem:HilbertSchmidt}
	Let $\Lambda\subset\R^{d}$ satisfy Assumption~\ref{ass:l}{\upshape (ii)} and let
	$V\in L^\infty(\R^d)$ have compact support in $[-R_{V},R_{V}]^{d}$ for some 
	$R_V >0$. Then, for every Fermi energy $E>0$ there exists a constant 
	$C_2 \equiv C_2(\Lambda,V,E)>0$ such that for all $L >0$ we have the bound
	\begin{equation}\label{eq:FermiProjDeff}
		\big\|1_{\Lambda_L^{c}}\big(1_{<E}(H_{0})-1_{<E}(H)\big)1_{\Lambda_L}\big\|_2 \le C_2.
	\end{equation}
\end{lemma}

\begin{proof}
We fix $E>0$.
To estimate the difference between the perturbed and the unperturbed Fermi projections we 
express them in terms of a contour integral as stated in Lemma~\ref{lem:H-Riesz}.
We set 
	\begin{equation}
 		\ell_{1}\equiv \ell_{1}(d,V,E) := \max_{z \in\image(\gamma)}\big\{\ell_{0}(d,V,z)\big\} < \infty,
	\end{equation}
where $\ell_{0}$ is defined in Lemma~\ref{lem:RoselventHilbertSchmidt} and
$\image(\gamma)$ denotes the image of the curve $\gamma$ in Lemma~\ref{lem:H-Riesz}.
We obtain for all $m,n\in \Z^d\setminus {}]-\ell_{1},\ell_{1}[\,^{d}$
	\begin{equation}\label{eq:Contour}
		1_{\Gamma_n}\big(1_{<E}(H_0)-1_{<E}(H)\big)1_{\Gamma_m}
		= -\frac{1}{2\pi\i} \,\oint_{\gamma}\dd z \, 1_{\Gamma_n}\Big(\frac{1}{H_0-z}-\frac{1}{H-z}\Big)1_{\Gamma_m}.
	\end{equation}
The Bochner integral exists even with respect to the Hilbert--Schmidt norm, as will follow from 
the estimates \eqref{eq:boundgama2} and \eqref{eq:boundgama1} below. 
We point out that \eqref{eq:boundgama1} relies again on the limiting absorption principle \eqref{eq:lim-new}.

In order to estimate the integral in \eqref{eq:Contour} we apply the resolvent identity twice to the integrand. 
The integrand then reads 
	\begin{equation}\label{eq:Integrand}
		1_{\Gamma_n}\Big(\frac{1}{H_0-z}V\frac{1}{H_0-z}-\frac{1}{H_0-z}V\frac{1}{H-z}V\frac{1}{H_0-z}\Big)1_{\Gamma_m}.
	\end{equation}
This implies the Hilbert--Schmidt-norm estimate
	\begin{multline}\label{eq:HoelderEstimate}
		\Big\|1_{\Gamma_n}
		\Big(\frac{1}{H_0-z}-\frac{1}{H-z}\Big)1_{\Gamma_m}\Big\|_2 \\
		\le \Big\|1_{\Gamma_n}\frac{1}{H_0-z}|V|^{1/2}\Big\|_4\Big(1+\Big\||V|^{1/2}\frac{1}{H-z}
				|V|^{1/2}\Big\|\Big)\Big\||V|^{1/2}\frac{1}{H_0-z}1_{\Gamma_m}\Big\|_4.
	\end{multline}
Lemma~\ref{lem:RoselventHilbertSchmidt} already provides bounds for the first and third factor on the right-hand side 
of \eqref{eq:HoelderEstimate}. To estimate the second factor we employ two different methods, depending on the 
location of $z$ on the contour. Therefore we split the curve $\gamma$ into two parts. We denote by $\gamma_{1}$ 
the right vertical part of $\gamma$ with image $\image(\gamma_{1}) = \big\{z\in\C:\,\Re z=E,\,|\Im z| \le \min\{E,1\} \big\}$. 
The remaining part of the curve $\gamma$ is denoted by $\gamma_2$. 

Let us first consider the curve $\gamma_2$. 
We observe 
	\begin{equation}
		\textrm{dist}\big(z,\sigma(H_{(0)})\big)\ge\min\{1,E\} \quad \text{for all } z\in\image(\gamma_{2}).
	\end{equation} 
Therefore, the middle factor in the second line of \eqref{eq:HoelderEstimate} is bounded from above by 
$(1+\|V\|_\infty/\min\{1,E\})$.
Since the curve $\gamma_2$ does not intersect $[0,\infty[$, there exists $\zeta_2 \equiv\zeta_2(V,E)>0$ 
such that $|\Im \sqrt z|/2\ge\zeta_2$ for all $z\in\image(\gamma_{2})\setminus\R$. 
Hence, according to Lemma \ref{lem:RoselventHilbertSchmidt} we estimate \eqref{eq:HoelderEstimate} by
	\begin{equation}\label{eq:boundgama2}
		\Big\|1_{\Gamma_n}\Big(\frac{1}{H_0-z}-\frac{1}{H-z}\Big)1_{\Gamma_m}\Big\|_2  \le
		\frac{c_{2} \,\e^{-\zeta_2(|n|+|m|)}}{(|n||m|)^{(d-1)/2}}
		\le \frac{c_{2}/\zeta_{2}}{(|n||m|)^{(d-1)/2}(|n|+|m|)}
	\end{equation}
for all $z\in \image(\gamma_2) \setminus \R$ with
	\begin{equation}
		c_2\equiv c_2(d,V,E):=C_1^2 \bigg(\max_{z \in\image(\gamma_{2})} |z|^{(d-3)/2}\bigg) 
		\bigg(1+\, \frac{\|V\|_\infty}{\min\{1,E\}}\bigg) < \infty.
	\end{equation}

We now turn our attention to $\gamma_1$, the part of the contour that intersects the continuous spectrum of $H$. 
Writing $\one = \Pi_{pp}(H) + \Pi_{c}(H)$ and recalling $\sigma_{pp}(H) \subset {}]-\infty,0]$, 
see the end of the proof of Lemma~\ref{lem:H-Riesz}, we infer 
	\begin{equation}\label{eq:LAusefull}
		\Big\||V|^{1/2}\frac{1}{H-z}|V|^{1/2}\Big\| \le \frac{\|V\|_\infty}{E}+\Big\||V|^{1/2}\frac{1}{H-z} 
			\Pi_{c}(H)|V|^{1/2}\Big\| 
	\end{equation}
for every $z\in\image(\gamma_{1}) \setminus\R$. The second term on the right-hand side admits the uniform upper bound 
\begin{equation}
	\label{eq:LAusefull2}
	\|\langle X\rangle |V|^{1/2}\|^2  \sup_{\substack{z\in\C:\,\Re z=E, \\ \Im z\neq 0}} 
					\Big\|\langle X\rangle^{-1}\frac{1}{H-z}\Pi_{c}(H)\langle X\rangle^{-1}\Big\|	
	\le (1+dR_V^2)\|V\|_\infty \,C_{LA}.
\end{equation}
Here, we used the compact support of $V$ and introduced the abbreviation $C_{LA} \equiv C_{LA}(d,E,V) < \infty$
for the supremum on the left-hand side of \eqref{eq:LAusefull2}. It is finite because of the limiting 
absorption principle \eqref{eq:lim-new} for $H$.

In addition, we need a lower bound for the decay rate of the exponential in \eqref{eq:4ResolvEst} along the curve
$\gamma_{1}$. We write $\image(\gamma_1) \ni z = E + \i \eta$ with $|\eta| \le \min\{1,E\}$. Then,  
	\begin{equation}\label{eq:sqrtzest}
		|\Im \sqrt z| =\sqrt[4]{E^2+\eta^2}\;\alpha(|\eta|/E)\ge\sqrt{E}\;\alpha(|\eta|/E),
	\end{equation} 
with $\alpha:\,[0, \infty[ {} \rightarrow [0,1]$, $x \mapsto \sin\big(\frac12 \arctan x\big)$. We note that 
$\sin y \ge y(1- y^{2}/6)$ for all $y \ge0$, $\arctan x \le \pi/2$ and $\arctan x \ge x/2$ for all $x \in [0,1]$.
Therefore, we infer the existence of 
a constant $\zeta_1\equiv\zeta_1(E)>0$ such that 
\begin{equation}
	\label{wurzel-z}
  |\Im \sqrt{z}|/2\ge\zeta_1|\eta| \qquad \text{for all } z=E+ \i\eta\in\image(\gamma_{1}).
\end{equation}
By applying Lemma \ref{lem:RoselventHilbertSchmidt} together with \eqref{wurzel-z}, as well as \eqref{eq:LAusefull} 
and \eqref{eq:LAusefull2}, we get the estimate 
	\begin{equation}\label{eq:boundgama1}
		\Big\|1_{\Gamma_n}\Big(\frac{1}{H_0-z}-\frac{1}{H-z}\Big)1_{\Gamma_m}\Big\|_2  \le 
		\frac{c_1\e^{-\zeta_1|\eta|(|n|+|m|)}}{(|n||m|)^{(d-1)/2}}
	\end{equation}
from \eqref{eq:HoelderEstimate} and any $\image(\gamma_1)\ni z=E+\i\eta$ with $|\eta|\le \min\{1,E\}$.
Here, we introduced the constant
	\begin{equation}
		c_1\equiv c_1(d,V,E):=C_1^2\Big(\max_{z\in\image(\gamma_1)}|z|^{(d-3)/2}\Big)
		\Big[1+\big(E^{-1}+(1+dR_V^2) C_{LA}\big)\|V\|_\infty\Big].
	\end{equation}
We are now able to estimate the contour integral in \eqref{eq:Contour} with the help of the 
bounds \eqref{eq:boundgama2} and \eqref{eq:boundgama1}
	\begin{align}\label{eq:FermiProjnm}
		\big\|1_{\Gamma_n}\big(1_{<E}(H_0)-1_{<E}(H)\big)1_{\Gamma_m}\big\|_2
		&	\le \frac{\tilde{c}_2}{(|n||m|)^{(d-1)/2}(|n|+|m|)}  \notag\\[1ex]
		&	\quad + \int_{-1}^1\dd\eta\,\frac{{c}_1\e^{-\zeta_1|\eta|(|n|+|m|)}}{2\pi(|n||m|)^{(d-1)/2}} \notag\\[1ex]
		& = \frac{\tilde{c}}{(|n||m|)^{(d-1)/2}(|n|+|m|)}
	\end{align}
for all $m,n\in\Z^d\setminus {}]-\ell_{1},\ell_{1}[\,^{d}$, where 
	\begin{equation}
%		\tilde{c}_1\equiv \tilde c_1(d,V,E)=\frac{c_1}{\pi\zeta_1},\qquad
		\tilde{c}_2\equiv \tilde c_2(d,V,E):=\frac{c_{2}(E+\|V\|_\infty+2)}{\pi\zeta_{2}} \quad
		\text{and}\quad \tilde{c}\equiv \tilde c(d,V,E):=\frac{c_1}{\pi\zeta_1} + \tilde{c}_2.
	\end{equation}

In order to prove the lemma for any $L>0$, we introduce a length $L_{0}>0$, which will be determined below, 
%%%In order to prove the lemma for any \red{$L>0$}, we recall that the origin is an interior point of the bounded 
%%%domain $\Lambda$, 
%%%whence there exists two radii $r_{\pm}>0$ such that 
%%%\begin{equation}
%%% 	B_{r_{-}} \subseteq \Lambda \subseteq B_{r_{+}}
%%%\end{equation}
%%%for the corresponding Euclidean balls $B_{r_{\pm}}$ which are centred about the origin. 
%%%In particular, there exists a length $L_{0} \equiv L_{0}(d,\Lambda,V,E)>0$ such that for all $L\ge L_{0}$
%%%\begin{equation}
%%%	\big|n\in\Z^{d}: \Gamma_{n} \cap \Lambda_{L}^{c} \neq\emptyset \big|\subset \Z^{d}\setminus [-\tilde{\ell}_{0},\tilde{\ell}_{0}]^{d}.
%%% \end{equation}
and first consider the case of $L \in {} ]0,L_0]$. 
In this case we have
	\begin{equation}\label{eq:LsmallerL0}
		\big\|1_{\Lambda_{L}^{c}}\big(1_{<E}(H_0)-1_{<E}(H)\big)1_{\Lambda_L}\big\|_2^{2}
		\le \big\|	\big(1_{<E}(H_0)-1_{<E}(H)\big)1_{\Lambda_{L_0}}\big\|_2^{2}.
	\end{equation}
Following \cite[Thm.\ B.9.2 and its proof]{artSEM1982Sim}, we infer the existence of a constant $C_{S}\equiv C_{S}(d,V,E)$ such that 
	\begin{equation}\label{eq:SeilerSimon}
		\big\|1_{<E}\big(H_{(0)}\big)1_{\Gamma_m}\big\|_1\le C_{S}
	\end{equation}
holds uniformly in $m\in\Z^d$. By applying the binomial inequality $(a+b)^2\le 2a^2+2b^2$ for $a,b\in\R$ and the inequality 
$\|A\|_2^2\le\|A\|_1$ for any trace-class operator $A$ with $\|A\|\le1$, we estimate the right-hand side of \eqref{eq:LsmallerL0} by
	\begin{multline}\label{eq:L0bound}
		2\Big(\big\|1_{<E}(H_0)1_{\Lambda_{L_0}}\big\|_2^2+\big\|1_{<E}(H)1_{\Lambda_{L_0}}\big\|_2^2\Big) \\
		\le \sum_{m\in \Xi_{L_{0}}} 2\big(\big\|1_{<E}(H_0)1_{\Gamma_m}\|_1
		+\big\|1_{<E}(H)1_{\Gamma_m}\big\|_1\big) \le 4 C_{S} |\tilde{\Lambda}_{L_{0}}| < \infty,
	\end{multline}
where we introduced the ``coarse-grained box domains'' 
\begin{equation}
 	\tilde{\Lambda}_{\ell}^{(\mathrm{ext})} := \bigcup_{m\in \Xi_{\ell}^{(\mathrm{ext})}} \Gamma_{m} \quad \text{with} \quad  
	\Xi_{\ell}^{(\mathrm{ext})}:=\big\{m\in \Z^d: \Gamma_{m} \cap \Lambda_\ell^{(c)} \neq \emptyset\big\} 
\end{equation}
for $\ell >0$. We note that $\tilde{\Lambda}_{\ell}^{\mathrm{ext}}$ is not the complement of 
$\tilde{\Lambda}_{\ell}$. It will be needed below.

In order to tackle the other case of $L>L_0$ we first determine a suitable value for $L_{0}$ as follows:
we recall that the origin is an interior point of the bounded 
domain $\Lambda$, 
whence there exists a length $L_{0} \equiv L_{0}(\Lambda,V,E)>0$ such that for all $L\ge L_{0}$
\begin{equation}
	\label{L0-cond}
	\tilde{\Lambda}_{L}^{\mathrm{ext}}	\subset \R^{d}\setminus {}]-\ell_{1},\ell_{1}[\,^{d}.
 \end{equation}
Now, we cover $\Lambda_L^c$ and $\Lambda_L \setminus \Lambda_{L_{0}}$ by unit cubes. 
Hence, we have 
	\begin{multline}\label{eq:QuadraticEntropy1}
		\big\|1_{\Lambda_L^{c}}\big(1_{<E}(H_0)-1_{<E}(H)\big)1_{\Lambda_L}\big\|_2^2 \\
		\le\big\|1_{\Lambda_{L}^{c}}\big(1_{<E}(H_0)-1_{<E}(H)\big)1_{\Lambda_{L_{0}}}\big\|_2^2 
		+\sum_{\substack{n\in \Xi_{L}^{\mathrm{ext}} \\[.5ex] m\in \Xi_{L}\cap \Xi_{L_{0}}^{\mathrm{ext}}}} 
		\big\|1_{\Gamma_n}
		\big(1_{<E}(H_0)-1_{<E}(H)\big)1_{\Gamma_m}\big\|_2^2.
\end{multline}
The first term on the right hand side of \eqref{eq:QuadraticEntropy1} is estimated by \eqref{eq:LsmallerL0} and \eqref{eq:L0bound}. To bound the double sum in \eqref{eq:QuadraticEntropy1} from above, we use \eqref{eq:FermiProjnm}, which is applicable due to the definition \eqref{L0-cond} of $L_{0}$, and obtain 	
	\begin{equation}\label{eq:sumbeforint}
	\big\|1_{\Lambda_L^{c}}\big(1_{<E}(H_0)-1_{<E}(H)\big)1_{\Lambda_L}\big\|_2^2 
	\le  4 C_{S} |\tilde{\Lambda}_{L_{0}}| +
		\sum_{\substack{n\in \Xi_{L}^{\mathrm{ext}} \\[.5ex] m\in \Xi_{L}\cap \Xi_{L_{0}}^{\mathrm{ext}}}}
		\frac{\tilde{c}^{2}}{(|n||m|)^{d-1} |n|^{2}}.
	\end{equation}
We conclude from the definition of $\ell_{1}$ that $|l|\ge |u|-\sqrt{d}\ge |u|/2$ 
for every $l\in \Xi_{L}^{\mathrm{ext}} \cup (\Xi_{L} \cap \Xi_{L_{0}}^{\mathrm{ext}})$ and every $u\in\Gamma_l \subseteq \R^{d} 
\setminus {}]-\ell_{1},\ell_{1}[\,^{d}$. Therefore we infer that the double sum in \eqref{eq:sumbeforint} 
is upper bounded by the double integral
	\begin{equation}\label{eq:IntermInteg}
				\int_{\tilde{\Lambda}_{L}}\dd x\int_{\tilde{\Lambda}_{L}^{\mathrm{ext}}}\dd y \;\frac{(2^{d}\tilde c)^2}{(|x||y|)^{d-1}|y|^{2}} 
		  = (2^{d}\tilde c)^2 \int_{\frac{L_{0}}{L}\tilde{\Lambda}_{L}} \frac{\dd x}{|x|^{d-1}} 
		  	\int_{\frac{L_{0}}{L}\tilde{\Lambda}_{L}^{\mathrm{ext}}} \frac{\dd y}{|y|^{d+1}}.
%		  = \int_{\frac{L_{0}}{L}\tilde{\Lambda}_{L}}\dd x\int_{\frac{L_{0}}{L}\tilde{\Lambda}_{L}^{\mathrm{ext}}}\dd y \;\frac{(2^{d}\tilde c)^2}{(|x||y|)^{d-1}|y|^{2}}.
	\end{equation}
But $\tilde{\Lambda}_{L}^{(\mathrm{ext})} \subseteq \bigcup_{x\in\Lambda_{L}^{(c)}} (x+ [-1,1]^{d})$ so that 
the scaled domains satisfy
\begin{equation}
 	\frac{L_{0}}{L} \,\tilde{\Lambda}_{L}^{(\mathrm{ext})} \subseteq \bigcup_{x\in\Lambda_{L_{0}}^{(c)}} 
	\Big(  x+ \frac{L_{0}}{L} [-1,1]^{d} \Big) \subseteq \bigcup_{x\in\Lambda_{L_{0}}^{(c)}} 
	\big(  x+ [-1,1]^{d} \big) =: K_{L_{0}}^{(\mathrm{ext})}
\end{equation}
for any $L \ge L_{0}$. Clearly, $K_{L_{0}}$ is bounded. Furthermore, we ensure 
that $K_{L_{0}}^{\mathrm{ext}}$ has a positive distance to the origin. This relies on the origin being an
interior point of $\Lambda$ and may require an enlargement of $L_{0}$, which can always be done. 
It follows that the right-hand side of \eqref{eq:IntermInteg} is bounded from above by some constant
$c_{3} \equiv c_{3}(\Lambda,V,E) <\infty$, uniformly in $L \ge L_{0}$.
Combining this with \eqref{eq:LsmallerL0}, \eqref{eq:L0bound}, \eqref{eq:sumbeforint} and \eqref{eq:IntermInteg}, we arrive at the final estimate
\begin{equation}
\sup_{L>0} \big\|1_{\Lambda_{L}^{c}}\big(1_{<E}(H_0)-1_{<E}(H)\big)1_{\Lambda_L}\big\|_2^{2}
 \le 4 C_{S} |\tilde{\Lambda}_{L_{0}}| + c_{3} =: C_{2}^{2}.
\end{equation}
\end{proof}

%%%%%%%%%%%%%%%%%%%%%%%%%%%%%%%%%%%%%%%%%%%%%%%%%%%%%%%%%%%%%%%%%
%%%%%%%%%%%%%%%%%%%%%%%%%%%%%%%%%%%%%%%%%%%%%%%%%%%%%%%%%%%%%%%%%
%
\subsection{Proof of the upper bound}\label{sec:upper}

We begin with an interpolation result.

%%%%%%%%%%%%%%%%%%%%%%%%%%%%%%%%%%%%%%%%%%%%%%%%%%Lemma
\begin{lemma}\label{lem:interpolation}
Let $\Lambda\subset\R^{d}$ be as in Assumption~\ref{ass:l}{\upshape (ii)}, let $V\in L^\infty(\R^d)$ have compact support and
fix $E>0$. Then there exists a constant $C_3\equiv C_3(\Lambda,V,E)>0$ such that for all $s\in\,]1/2,1[$ and all $ L\in\N$ we have
	\begin{equation}\label{eq:FermiProjDivs}
		\big\|1_{\Lambda_L^{c}}\big(1_{<E}(H)-1_{<E}(H_0)\big)1_{\Lambda_L}\big\|_{2s}^{2s}\le C_3 L^{2d(1-s)}.
	\end{equation}
\end{lemma}
\begin{proof}
Given a trace-class operator $A$ and  $s\in{}]1/2,1[\,$, we conclude from the interpolation inequality,  
see e.g.\ \cite[Lemma 1.11.5]{MR2760403},
	\begin{equation}
		\|A\|_{2s}^{2s}\le\|A\|_1^{2(1-s)}\|A\|_2^{2(2s-1)}. %\le \|A\|_1^{2(1-s)}(\|A\|_2^{2}+1).
	\end{equation}
The estimate \eqref{eq:SeilerSimon} implies that the operator 
\begin{equation}
	\label{eq:AL-def}
 	A_L:=1_{\Lambda_L^{c}}\big(1_{<E}(H)-1_{<E}(H_0)\big)1_{\Lambda_L}
\end{equation}
is trace class for all $L\in\N$ with norm 
$\|A_L\|_1\le 2(2L)^dC_{S}$. Moreover, $\|A_L\|_2^2\le C_2^2$ for all $L\in\N$ by Lemma \ref{lem:HilbertSchmidt}.
This proves the claim with
	\begin{equation}
		2^{d+1}C_{S}C_2^{2(2s-1)} \le 2^{d+1}C_{S}(C_2^{2}+1) =: C_3\equiv C_3(\Lambda,V,E).
	\end{equation}
\end{proof}

%%%%%%%%%%%%%%%%%%%%%%%%%%%%%%%%%%%%%%%%%%%%%%%%%Remark
\begin{remark}\label{ReM.Abkuerzung}
Lemma~\ref{lem:interpolation} allows for a quick proof of the upper bound in Theorem~\ref{thm:main}, if we restrict ourselves to the case $d\ge 2$. First, we apply the upper bound in \eqref{h-bounds} to the entanglement entropy and rewrite it with \eqref{eq:gofProj} 
	\begin{align}
		\label{eq:no-s}
		S_E(H,\Lambda_L) & \le \frac{6}{1-s}\big\|1_{\Lambda_L^{c}}1_{<E}(H)1_{\Lambda_L} \big\|_{2s}^{2s} \notag\\
		& \le \frac{12}{1-s}\Big(\big\|1_{\Lambda_L^{c}}1_{<E}(H_0)1_{\Lambda_L}
		\big\|_{2s}^{2s}+\|A_L\|^{2s}_{2s}\Big).
	\end{align}
	Here, $A_L$ is defined in \eqref{eq:AL-def}. The first term on the right-hand side scales like 
	$\mathcal O(L^{d-1}\ln L)$ according to the lemma and subsequent remarks in \cite{LeschkeSobolevSpitzer14}. 
	The second term is of order $\mathcal O\big(L^{2d(1-s)}\big)$ according to Lemma~\ref{lem:interpolation}.  If we choose $s\equiv s(d,\epsilon):=1-\epsilon(2d)^{-1}$ for any $\epsilon\in[0,1]$ the second term is of the order $\mathcal O(L^\epsilon)$, and thus subleading as compared to the first term in all but one dimensions.

	Unfortunately, there is no choice for $s$ which yields only a logarithmic growth in $d=1$. 
	To appropriately bound the term $(1-s)^{-1} \mathcal O\big(L^{2d(1-s)}\big)$ in \eqref{eq:no-s} 
	requires an $L$-dependent choice of $s$
	with $s\equiv s(L) \rightarrow 1$ as $L\to\infty$. However, such a choice of $s$ leads to an additional 
	diverging prefactor $(1-s)^{-1}$ multiplying the asymptotics $\mathcal O(L^{d-1}\ln L)$ from the first term.
\end{remark}

%%%%%%%%%%%%%%%%%%%%%%%%%%%%%%%%%%%%%%%%%%%%%%%%%%
We now present an approach, which yields the optimal upper bound of order 
$\mathcal O(L^{d-1}\ln L)$ for all dimensions.
%%%%%%%%%%%%%%%%%%%%%%%%%%%%%%%%%%%%%%%%%%%%%Lemma
\begin{lemma}\label{lem:logtriangleineq}
Let $A$ and $B$ be two compact operators with $\|A\|,\|B\|\le \e^{-1/2}/3$ and consider the function
\begin{equation}
 	f: \,[0,\infty[ \, \rightarrow [0,1], \; x \mapsto -1_{[0,1]}(x)\,x^2\log_{2} (x^2).
\end{equation}
Then we have
	\begin{equation}
		\tr\{f(|A|)\}\le 4 \tr\{f(|B|)\} + 4 \tr\{f(|A-B|)\}.
%		-\tr\big\{|A|^2\ln |A|^2\big\}\le-4\tr\big\{|B|^2\ln |B|^2\big\}-4 \tr\big\{|A-B|^2\ln |A-B|^2\big\}.
	\end{equation}
\end{lemma}
%%%%%%%%%%%%%%%%%%%%%%%%%%%%%%%%%%%%%%%%%%%%%%

For any compact operator $A$ let $\big( a_n(A)\big)_{n\in\N}\subseteq[0,\infty[{}$ denote the non-increasing sequence
of its singular values. They coincide with the eigenvalues of the self-adjoint operator $|A|$.

%%%%%%%%%%%%%%%%%%%%%%%%%%%%%%%%%%%%%%%%%%%%%%
\begin{proof}[Proof of Lemma \ref{lem:logtriangleineq}]
	By assumption, we have $0 \le a_{2n}(A) \le a_{2n-1}(A) \le \e^{-1/2}/3$ for all $n\in\N$. 
	Since the function $f$ is monotonously increasing on $[0, \e^{-1/2}]$, we deduce
	\begin{equation}\label{eq:Asum}
		\tr\{f(|A|)\}=\sum_{n\in \N}f\big(a_n(A)\big)\le 2\sum_{n\in \N}f\big(a_{2n-1}(A)\big).
	\end{equation}
The singular values of any compact operators $A$ and $B$ satisfy the inequality
	\begin{equation}\label{eq:singularvalues}
		a_{n+m-1}(A)\le a_n(B)+a_m(A-B)
	\end{equation}
for all $n,m\in\N$ \cite[Prop.\ 2 in Sect.\ III.G]{MR1144277}. We point out that the right-hand side 
of \eqref{eq:singularvalues} does not exceed the upper bound $\e^{-1/2}$ because of $\|A-B\| \le 
\|A\| + \|B\| \le (2/3)\e^{-1/2}$. 
Together with the monotonicity of $f$, we conclude 
from \eqref{eq:Asum} that 
	\begin{equation}\label{eq:Asum2}
		\tr\{f(|A|)\} \le  2\sum_{n\in \N}f\big(a_{n}(B) + a_{n}(A-B)\big).
	\end{equation}
Next, we claim that 
	\begin{equation}\label{eq:LogIneq}
		f(x+y)\le -2(x^2+y^2)\log_2 [(x+y)^2] \le 2f(x)+2f(y)
	\end{equation}
for all $x,y\ge0$ with $x+y<1$. The first estimate follows from the binomial inequality together with 
$-\log_{2}[(x+y)^{2}] \ge 0$ for $x+y <1$, the second estimate from $(x+y)^{2} \ge x^{2}$, respectively $(x+y)^{2} \ge y^{2}$, and the fact that $-\log_2$ 
is monotonously decreasing.
Combining \eqref{eq:Asum2} and \eqref{eq:LogIneq}, we arrive at
	\begin{equation}
		\tr\{f(|A|)\} \le  4\sum_{n\in \N} \big[ f\big(a_n(B)\big)+f\big(a_n(A-B)\big) \big].
	\end{equation}
\end{proof}

%%%%%%%%%%%%%%%%%%%%%%%%%%%%%%%%%%%%%%%%%%%%Theorem
\begin{proof}[Proof of the upper bound in Theorem \ref{thm:main}]
Let $L >0$ and $E >0$. Lemma \ref{lem:logquatrat} and \eqref{eq:gofProj} yield 
	\begin{equation}\label{eq:firststep}
		S_E(H,\Lambda_L)\le 3\sum_{n=1}^\infty f\big(a_n(1_{\Lambda_L^{c}}1_{<E}(H)1_{\Lambda_L})\big),
	\end{equation}
where $f$ was defined in Lemma~\ref{lem:logtriangleineq}.	In order to apply Lemma \ref{lem:logtriangleineq}, we
will decompose the compact operator $1_{\Lambda_L^{c}}1_{<E}\big(H_{(0)}\big)1_{\Lambda_L}$ into a part bounded by 
${\e}^{-1/2}/3$ in norm and a finite-rank operator. To this end, we introduce
	\begin{equation}
		N_{(0)}\equiv N_{(0)}(\Lambda,V,E,L):=
		\min\Big\{n\in\N:\;a_{n}\big(1_{\Lambda_L^{c}}1_{<E}(H_{(0)})1_{\Lambda_L}\big)\le{\e}^{-1/2}/3 \Big\}-1,
	\end{equation}
the number of singular values of $1_{\Lambda_L^{c}}1_{<E}\big(H_{(0)}\big)1_{\Lambda_L}$ which are larger than
${\e}^{-1/2}/3$. We define $F_{(0)}$ as the contribution from the first $N_{(0)}$ singular values in the 
singular-value decomposition of $1_{\Lambda_L^{c}}1_{<E}\big(H_{(0)}\big)1_{\Lambda_L}$, whence 
$\rank (F_{(0)}) = N_{(0)}$ and $\|F_{(0)}\|\le1$. The remainder
	\begin{equation}
		\label{A+F}
		Q_{(0)} := 1_{\Lambda_L^{c}}1_{<E}(H_{(0)})1_{\Lambda_L} - F_{(0)}
	\end{equation}
fulfils $\|Q_{(0)}\|\le{\e}^{-1/2}/3$ by definition of $N_{(0)}$.
We note the upper bound 
	\begin{equation}
		N_{(0)} \le 9\e \sum_{n=1}^{N_{(0)}} \Big(a_{n}\big(1_{\Lambda_L^{c}}1_{<E}(H_{(0)})1_{\Lambda_L}\big)\Big)^{2}
		\le 9\e\big\|1_{\Lambda_L^{c}}1_{<E}(H_{(0)})1_{\Lambda_L}\big\|_2^2.
	\end{equation}
Using Lemma \ref{lem:HilbertSchmidt}, we further estimate $N$ in terms of unperturbed quantities
	\begin{equation}\label{eq:NB}
		N\le 18\e\big\|1_{\Lambda_L^{c}}1_{<E}(H_{0})1_{\Lambda_L}\big\|_2^2 + 18\e C_2^2.
	\end{equation}
The identity \eqref{eq:gofProj} and the lower bound in \eqref{eq:hEstimate} imply 
$\|1_{\Lambda_L^{c}}1_{<E}(H_{0})1_{\Lambda_L}\|_2^2 \le S_{E}(H_{0},\Lambda_{L})$ so that we  
obtain 
	\begin{equation}
 		N_{0} \le 9\e S_{E}(H_{0},\Lambda_{L})  \quad \text{and}\quad
		N \le 18 \e S_{E}(H_{0},\Lambda_{L}) + 18 \e C_{2}^{2} \label{bound:N-S}
	\end{equation}
for later usage.

We deduce from \eqref{eq:singularvalues} and $\rank(F)=N$ that for all $n\in\N$
	\begin{equation}\label{eq:NAeigenest}
		a_{n+N}(Q + F)\le a_n (Q)+a_{N+1}(F)
		=a_n(Q)\le \e^{-1/2}/3.
	\end{equation}
Hence, \eqref{eq:firststep} implies that 
	\begin{equation}\label{eq:secondstep}
		S_{E}(H,\Lambda_{L}) \le 3\sum_{n=1}^{N}f\big(a_n(Q+ F)\big)+3\sum_{n=1}^\infty 
		f\big(a_n(Q)\big)\le 3N + 3\tr\{f(|Q|)\},
	\end{equation}
where used the monotonicity of $f$ on ${[0,\e^{-1/2}]}$ and $f\le 1$.
Now, Lemma \ref{lem:logtriangleineq} allows to estimate \eqref{eq:secondstep} so that 
	\begin{equation}\label{eq:FofDifA}
		S_E(H,\Lambda_L) \le 3N + 12 \tr\{f(|Q_{0}|)\} + 12 \tr\{f(|\delta Q|)\},
	\end{equation} 
where $\delta Q :=Q-Q_{0}$. The rank of $\delta F := F - F_{0}$
obeys 
	\begin{equation}
		\label{bound:M}
		\deltaN \equiv \deltaN (\Lambda,V,E,L):=\rank(\delta F)\le N + N_{0}.
	\end{equation} 
We deduce again from  \eqref{eq:singularvalues} and from the definition of $\deltaN$ that for all $n\in\N$
	\begin{equation}
		\label{inrange}
		a_{n+2\mkern1mu\deltaN}(\delta Q) = a_{(n+\deltaN)+(\deltaN+1)-1}(\delta Q)\le a_{n+\deltaN}(\delta Q + \delta F).
	\end{equation}
Yet another application of \eqref{eq:singularvalues} and the definition of $\deltaN$ yield for all $n\in\N$
	\begin{equation}
		a_{n+\deltaN}(\delta Q + \delta F) \le a_{n}(\delta Q) \le \|\delta Q\| \le 2 \e^{-1/2} /3 .
	\end{equation}
Therefore the singular values in \eqref{inrange} lie in the range where the function $f$ is monotonously 
increasing. Hence, we obtain
	\begin{align}
		\label{eq:EstDiffA}
		\tr\{f(|\delta Q|)\} &\le \sum_{n=1}^{2\mkern1mu\deltaN}f\big(a_n(\delta Q)\big)
		+\sum_{n\in\N}f\big(a_{\deltaN+n}(\delta Q + \delta F)\big) \notag \\
		& \le 2\,\deltaN + \sum_{n\in\N}f\big(a_{n}(\delta Q + \delta F)\big), 
	\end{align}
	where the second line follows from $0 \le f \le 1$.

Now, we repeat the arguments from \eqref{inrange} to \eqref{eq:EstDiffA} for $Q_{0}$ instead of $\delta Q$,
$F_{0}$ instead of $\delta F$ and $N_{0}$ instead of $\deltaN$. This implies
	\begin{equation}
		\label{n-sum}
 		\tr\{f(|Q_{0}|)\} \le 2 N_{0} + \sum_{n\in\N}f\big(a_{n}(Q_{0} + F_{0})\big).
	\end{equation}
The sum in \eqref{n-sum} is bounded from above by the unperturbed entanglement entropy, which follows from \eqref{A+F}, 
the definition of $f$, \eqref{eq:gofProj} 
and the lower bound in Lemma \ref{lem:logquatrat}, whence
	\begin{equation}
		\label{AL0}
 		\tr\{f(|Q_{0}|)\} \le 2 N_{0} + S_{E}(H_{0},\Lambda_L).
	\end{equation}
Next, we combine \eqref{eq:FofDifA}, \eqref{bound:N-S}, \eqref{eq:EstDiffA}, \eqref{bound:M} and \eqref{AL0} 
to obtain
	\begin{equation}
		\label{upper-prefinal}
 		S_{E}(H,\Lambda_{L}) \le 2508 \, S_{E}(H_{0},\Lambda_{L}) + 1322\, C_{2}^{2} 
			+ 12 \sum_{n\in\N}f\big(a_{n}(\delta Q+ \delta F)\big).
	\end{equation}
In order to estimate the sum in \eqref{upper-prefinal}, we appeal to the definitions of 
$\delta Q$ and $\delta F$, \eqref{A+F}, the definition of $f$ and \eqref{eq:xlogxestxs} to deduce 
	\begin{equation}
		\sum_{n\in\N}f\big(a_{n}(\delta Q+ \delta F)\big) \le\frac{1}{1-s}\, \big\|1_{\Lambda_L^{c}}\big(1_{<E}(H_{0})-1_{<E}(H)\big)1_{\Lambda_L}\big\|_{2s}^{2s}
	\end{equation}
for any $s\in {}]0,1[\,$. Restricting ourselves to $s\in {}]1/2,1[\,$ allows us to apply 
Lemma~\ref{lem:interpolation} so that 
	\begin{equation}
		\sum_{n\in\N}f\big(a_{n}(\delta Q+ \delta F)\big) \le \frac{C_3}{1-s} \, L^{2d(1-s)},
	\end{equation}
where $C_3=C_3(\Lambda,V,E)>0$ is given in Lemma \ref{lem:interpolation} and independent of $s$. 
Assuming $L \ge 8$, we choose the $L$-dependent exponent
	\begin{equation}
		s\equiv s(L):=1-\frac{1}{\ln L}\in{}]1/2,1[{}
	\end{equation}
	which implies
	\begin{equation}
		\label{Differenz-2s}
		\sum_{n\in\N}f\big(a_{n}(\delta Q+ \delta F)\big) \le C_3\e^{2d}\ln L.
	\end{equation}
 The entanglement entropy of a free Fermi gas exhibits an enhanced area law, 
 $S_E(H_0,\Lambda_L)=\mathcal O(L^{d-1}\ln L)$ \cite[Theorem]{LeschkeSobolevSpitzer14}, so that the claim follows from 
 \eqref{upper-prefinal} together with \eqref{Differenz-2s}. 
\end{proof}

%%%%%%%%%%%%%%%%%%%%%%%%%%%%%%%%%%%%%%%%%%%%%%%%%%%%%%%%%%%%%%%%%
%%%%%%%%%%%%%%%%%%%%%%%%%%%%%%%%%%%%%%%%%%%%%%%%%%%%%%%%%%%%%%%%%
%
\subsection{Proof of the lower bound}
\label{sec:lower}

%%%%%%%%%%%%%%%%%%%%%%%%%%%%%%%%%%%%%%%%%%%%%%%%%%Theorem
\begin{proof}[Proof of the lower bound in Theorem \ref{thm:main}] 
We fix $L>0$ and $E>0$.
The lower bound in \eqref{eq:hEstimate}, the identity  \eqref{eq:gofProj} and the elementary inequality $(a-b)^2\ge a^2/2 -b^2$ for $a,b\in\R$ imply
	\begin{align}
		S_E(H,\Lambda_L)&\ge\tr\big\{g\big(1_{\Lambda_L}1_{<E}(H)1_{\Lambda_L}\big)\big\}=\|1_{\Lambda_L^{c}}
		1_{<E}(H)1_{\Lambda_L}\|_2^2\notag\\
		&\ge\frac{1}{2}\big\|1_{\Lambda_L^{c}}1_{<E}(H_0)1_{\Lambda_L}\big\|^2_2-\big\|1_{\Lambda_L^{c}}
		\big(1_{<E}(H_0)-1_{<E}(H)\big)1_{\Lambda_L}\big\|^2_2.
	\end{align}
	The second term on the right-hand side is uniformly bounded in $L$ according to Lemma \ref{lem:HilbertSchmidt}.
	For the first term, it was shown in \cite[Eq. (7)]{LeschkeSobolevSpitzer14} that the leading behaviour of the 
	asymptotic expansion in $L$ is of order $L^{d-1}\ln L$. Hence,
	\begin{equation}
		\label{constant-lower}
		\liminf_{L\rightarrow\infty}\frac{S_E(H,\Lambda_L)}{L^{d-1}\ln L}\ge \frac{1}{2} \;
		\lim_{L\rightarrow\infty}\frac{\tr\big\{g\big(1_{\Lambda_L}
		1_{<E}(H_{0})1_{\Lambda_L}\big)\big\}}{L^{d-1}\ln L} =: \Sigma_{l}.
	\end{equation}
	Finally, Eqs.\ (1), (4), (7) and (8) in \cite{LeschkeSobolevSpitzer14} and \eqref{eq:widom} imply 
	\begin{equation}
 		\label{sigmal-def}
		\Sigma_{l} = \frac{3}{2\pi^{2}}\; \Sigma_{0}.
\end{equation}
\end{proof}

%%%%%%%%%%%%%%%%%%%%%%%%%%%%%%%%%%%%%%%%%%%%%%%%%%%%%%%%%%%%%%%%%
%%%%%%%%%%%%%%%%%%%%%%%%%%%%%%%%%%%%%%%%%%%%%%%%%%%%%%%%%%%%%%%%%
%
\appendix
\section{Auxiliary results}
%
%%%%%%%%%%%%%%%%%%%%%%%%%%%%%%%%%%%%%%%%%%%%%%%%%%%%%%%%%%%%%%%%%
%%%%%%%%%%%%%%%%%%%%%%%%%%%%%%%%%%%%%%%%%%%%%%%%%%%%%%%%%%%%%%%%%

The following representation \eqref{eq:riesz-rep} of the Fermi projection in terms of a Riesz projection 
with the integration contour cutting through the continuous spectrum may be of independent interest.

\begin{theorem}
	\label{thm:abs-Riesz}
 	Let $K$ be a densely defined self-adjoint operator in a Hilbert space $\mathcal{H}$, which is bounded below and
	satisfies a limiting absorption principle at $E \in \R$ in the sense that there exists a bounded 
	operator $B$ on $\mathcal{H}$ with inverse $B^{-1}$, which is possibly only densely defined and unbounded, such that
	\begin{equation}
		\label{eq:lim-abs-new}
		\mathcal{S}_{E} :=\sup_{z\in\C:\,\Re z=E, \, \Im z\neq 0}
		\Big\| B \frac{1}{K-z}\Pi_{c}(K) B \Big\| < \infty .
	\end{equation}
	Here, $\Pi_{c}(K)$ denotes the projection onto the continuous spectral subspace of $K$. 
	Let $A_{1}, A_{2}$ be two bounded operators on $\mathcal{H}$ such that $\|A_{1}B^{-1}\| <\infty$ 
	and $\| B^{-1} A_{2}\| <\infty$. 
	Finally, we assume that there are no eigenvalues of $K$ near $E$, i.e.\
	$\dist\big(\sigma_{pp}(K), E\big) >0$. Then we have the representation
	\begin{equation}
		\label{eq:riesz-rep}
 		A_{1} 1_{<E}(K) A_{2} = - \frac{1}{2\pi\i} \oint_{\gamma} \d z\, A_{1} \,\frac{1}{K-z} \, A_{2}.
	\end{equation}
	The right-hand side of \eqref{eq:riesz-rep} exists as a Bochner integral with respect to the operator norm 
	$\|\pmb\cdot\|$, and the integration contour $\gamma$ is a closed curve in 
	in the complex plane $\C$ which, for $s>0$, traces the boundary of the rectangle 
	$\big\{z\in\C:\;|\Im z|\le s, \;\Re z\in[-1+ \inf\sigma(K),E] \big\}$
	once in the counter-clockwise direction.
\end{theorem}

\begin{proof}
	Let $\varepsilon >0$ and let $\gamma_{\varepsilon}$ be the curve $\gamma$ without the vertical line segment 
	from $E- \i \varepsilon$  to $E+ \i\varepsilon$. Since $\|(K-z)^{-1}\|$ is uniformly bounded for 
	$z$ in the image of $\gamma_{\varepsilon}$, it suffices to verify that
	\begin{equation}
		\label{eq:Bochner}
		\int_{-\varepsilon}^{\varepsilon} \d\eta \; \Big\|  A_{1} \,\frac{1}{K-E -\i\eta} \, A_{2}\Big\| < \infty
	\end{equation}
	in order to show the existence of the right-hand side of \eqref{eq:riesz-rep} as a Bochner integral 
	with respect to the operator norm. 
	But 
	\begin{align}
 		\Big\|  A_{1} \,\frac{1}{K-E -\i\eta} \, A_{2}\Big\| 
		& \le  \Big\|  A_{1} \,\frac{1}{K-E -\i\eta} \,\Pi_{pp}(K)\, A_{2}\Big\|  \notag \\
		& \quad + \| A_{1}B^{-1}\| \|B^{-1}A_{2} \| \Big\| B \frac{1}{K-E -\i\eta}\,\Pi_{c}(K) B \Big\| \notag \\
		& \le \frac{\|A_{1}\| \|A_{2}\|}{\dist\big(\sigma_{pp}(K),E\big)} + \| A_{1}B^{-1}\| \|B^{-1}A_{2} \|
			\,\mathcal{S}_{E}
	\end{align}
	uniformly in $\eta\in[-\varepsilon,\varepsilon]$, and the estimate \eqref{eq:Bochner} holds.  
	
	It remains to prove the equality in \eqref{eq:riesz-rep}. Let $\varphi,\psi \in\mathcal{H}$. Since the contour integral along $\gamma$ exists in the Bochner sense with respect to the operator norm, we equate
	\begin{align}
		\label{eq:abs-riesz-prefinal}
 		\bigg\langle\varphi,  \bigg(\oint_{\gamma} \d z\, A_{1} \,\frac{1}{K-z} \, A_{2} \bigg) \psi\bigg\rangle
		& = \lim_{\varepsilon\searrow 0} \int_{\gamma_{\varepsilon}}
		    \d z\, \langle\varphi, A_{1} \,\frac{1}{K-z} \, A_{2} \psi\rangle \notag\\
		& = \lim_{\varepsilon\searrow 0} \int_{\R} \d\mu_{(A_{1}^{*}\varphi), (A_{2}\psi)}(\lambda) \,
				\int_{\gamma_{\varepsilon}} \d z \, \frac{1}{\lambda-z} ,
	\end{align}
	where we introduced the complex spectral measure $\mu_{\varphi,\psi} := \langle\varphi, 1_{\bullet}(K)\psi\rangle$
	of $K$ and used Fubini in the last step. On the other hand, we apply the residue theorem to conclude
	\begin{equation}
		\label{eq:abs-riesz-final}
 		-2\pi\i \,	\langle\varphi, A_{1} 1_{<E}(K) A_{2}\psi\rangle
		=  \int_{\R} \d\mu_{(A_{1}^{*}\varphi), (A_{2}\psi)}(\lambda) \,\int_{\gamma} \d z \, \frac{1}{\lambda-z},
	\end{equation}
	which is justified because $E$ is not an eigenvalue of $K$. The right-hand side of \eqref{eq:abs-riesz-final} 
	equals
	\begin{equation}
		\label{eq:abs-riesz-fifinal}
		\lim_{\varepsilon\searrow 0}  \int_{\R} \d\mu_{(A_{1}^{*}\varphi), (A_{2}\psi)}(\lambda) \,\int_{\gamma_{\varepsilon}} \!\!\d z \, \frac{1}{\lambda-z} 
		 +\i \lim_{\varepsilon\searrow 0}  \int_{\R} \d\mu_{(A_{1}^{*}\varphi), (A_{2}\psi)}(\lambda)
			\int_{-\varepsilon}^{\varepsilon}\!\d\eta\, \frac{1}{\lambda -E -\i\eta}.
	\end{equation}
	The explicit computation, using symmetry, 
	\begin{equation}
 		\int_{-\varepsilon}^{\varepsilon}\d\eta\, \frac{1}{\lambda -E -\i\eta}
		= \int_{-\varepsilon}^{\varepsilon}\d\eta\, \frac{\lambda-E}{(\lambda -E)^{2} +\eta^{2}}
		= 2 \arctan\Big( \frac{\varepsilon}{\lambda-E}\Big)
	\end{equation}
	holds for every real $\lambda\neq E$. Therefore, dominated convergence implies that the second limit in \eqref{eq:abs-riesz-fifinal} vanishes. Here, we used again that $E$ is not an eigenvalue of $K$. 
	Since $\varphi$ and $\psi$ are arbitrary, the theorem follows from \eqref{eq:abs-riesz-prefinal} to
	\eqref{eq:abs-riesz-fifinal}.
\end{proof}

\begin{remark}
 	Theorem~\ref{thm:abs-Riesz} readily generalises from Fermi projections to spectral projections of more 
	general intervals.
\end{remark}

In the remaining part we prove some elementary estimates.

%%%%%%%%%%%%%%%%%%%%%%%%%%%%%%%%%%%%%%%%%%%%%%%%%%Lemma
\begin{lemma}\label{lem:gs}
For all $s\in\,]0,1[$  and all $x\in[0,1]$ we have
	\begin{equation}\label{eq:xlogxestxs}
		-x\log_{2} x \le \frac{x^{s}}{1-s}.
	\end{equation}
and 
	\begin{equation}\label{eq:hEstimate}
		g(x) \le h(x)\leq \frac{6}{1-s} \big(g(x)\big)^s,
	\end{equation}
where $g$ was defined in \eqref{eq:g}.
\end{lemma}
\begin{proof}
We introduce the continuous function $\varphi : [0,1] \rightarrow [0,\infty[{},\; x \mapsto -x^{1-s}\log_{2} x$. The fist claim follows from the  observation
	\begin{equation}
		\label{eq:phi-ub}
		0\le \varphi\le \frac{1}{1-s},
	\end{equation}
which holds true because $\varphi(1) = \varphi(0) = 0$ and $\varphi$ 
has a unique maximum at $\e^{-1/(1-s)}$.
 
Due to the symmetry $h(x)=h(1-x)$ and $g(x)=g(1-x)$ for all $x\in[0,1]$ it is sufficient to prove 
\eqref{eq:hEstimate} for all $x\in[0,1/2]$ only. 
As for the upper bound in \eqref{eq:hEstimate}, we note that with $\psi : [0,1/2] \rightarrow [0,\infty[{},\; x \mapsto -(1-x)\log_{2}(1-x)$, we have 
	\begin{equation}
		\psi(x) \le \frac{x}{\ln 2} \le \frac{x^{s}}{\ln 2} \quad \text{for all } x\in[0,1/2],
	\end{equation}
because $\psi(0)=0$ and $\psi' \le 1/\ln2$. This and \eqref{eq:phi-ub} imply 
	\begin{equation}
		h(x) = x^{s} \varphi(x) +\psi(x) \le  x^s \Big(\frac{1}{\ln 2}+\frac{1}{1-s}\Big) \le \frac{6}{1-s} \, 
		\big(x(1-x)\big)^{s}
	\end{equation}
for all $x\in[0,1/2]$.

The argument for the lower bound is similar to the above. Since $h(0)=g(0)=0$ it suffices to show $h' \ge g'$ on ${}]0,1/2]$.
We observe $h'(1/2) = g'(1/2) =0$, introduce $\gamma(y):= g'(-y + 1/2)$, $\eta(y) := h'(-y + 1/2)$ for 
$y\in [0, 1/2[$ and verify $\eta' \ge 2 = \gamma'$. This yields the claim. 
\end{proof}
%%%%%%%%%%%%%%%%%%%%%%%%%%%%%%%%%%%%%%%%%%%%%

%%%%%%%%%%%%%%%%%%%%%%%%%%%%%%%%%%%%%%%%%%%%Lemma
\begin{lemma}\label{lem:logquatrat}
For every $x\in[0,1]$ we have
	\begin{equation}
		-g(x)\log_{2} g(x)\le h(x)\le -3g(x)\log_{2} g(x).
	\end{equation}
\end{lemma}
\begin{proof}
Since $g(x)\le\min\{x,1-x\}$ for all $x\in[0,1]$, the left inequality of the claim follows from 
	\begin{equation}
		-g(x)\log_{2} g(x)=-g(x)\big(\log_{2} x+\log_{2} (1-x)\big)\le h(x).
	\end{equation}
For the right inequality we consider only $x\in[0,1/2]$, which suffices by symmetry. We rewrite
	\begin{equation}
		\label{eq:A13upper}
		-3g(x)\log_{2} g(x)-h(x)=-x p(x)\log_{2}x - q(x)\log_{2}(1-x)
	\end{equation}
with $p(x):=2-3x$ and $q(x):=-1+4x-3x^2$.
The polynomial $q$ is negative on the interval $[0,1/3[$ and positive on $]1/3,1/2]$ while $p$ is positive everywhere on $[0,1/2]$. Therefore for all $x\in[1/3,1/2]$ we have
	\begin{equation}
		-3g(x)\log_{2} g(x)-h(x)\ge0.
	\end{equation}
On the other hand, we claim that  
	\begin{equation}
		\label{eq:concave}
		\log_{2} (1-x) \ge 2x \log_{2}x
	\end{equation}
for all $x\in[0,1/3]$ because the function $[0,1/2] \ni x \mapsto -2x \log_{2}x + \log_{2}(1-x)$ vanishes at $x=0$
and at $x=1/2$ and it is concave. Therefore it must be non-negative.
Inserting \eqref{eq:concave} into \eqref{eq:A13upper}, we obtain   
	\begin{equation}
		-3g(x)\log_{2} g(x)-h(x)\ge -x(\log_{2} x)\big(p(x)+2q(x)\big)\ge 0
	\end{equation}
because $p(x)+2q(x)=5x-6x^2\ge0$ for all $x\in[0,1/3]$.
\end{proof}

\section*{Acknowledgement}
We thank Wolfgang Spitzer (FU Hagen) for comments which helped improve a prior version of this paper.

%%%%%%%%%%%%%%%%%%%%%%%%%%%%%%%%%%%%%%%%%%%%%%%%%%%%%%%%%%%%%%%%%
%%%%%%%%%%%%%%%%%%%%%%%%%%%%%%%%%%%%%%%%%%%%%%%%%%%%%%%%%%%%%%%%%
%

\end{document}